\documentclass[conference]{IEEEtran}
\usepackage{amsmath,amssymb,amsthm,bm}
\usepackage{booktabs,graphicx}
\usepackage{tikz}
\usetikzlibrary{arrows.meta}
\usepackage{algorithm,algpseudocode}

\makeatletter
\renewcommand{\fst@algorithm}{%
  \fs@ruled
  \def\@fs@pre{%
    \kern0.10in
    \hrule height.8pt depth0pt
    \kern2pt
  }%
}
\makeatother
\usepackage{cite}
\usepackage[hidelinks]{hyperref}
\hypersetup{pdftitle={Low-Latency Belief-Propagation Decoding of Quantum LDPC Codes with Reinforcement-Learned Scheduling},
 pdfauthor={Mohsen Moradi, Ghilas Aissou, Taejoon Kim, Remi Chou, David Mitchell}}
\newtheorem{proposition}{Proposition}

\newcommand{\cB}{\mathcal{B}}
\newcommand{\cR}{\mathcal{R}}
\newcommand{\cV}{\mathcal{V}}
\newcommand{\cN}{\mathcal{N}}
\newcommand{\cC}{\mathcal{C}}
\newcommand{\cA}{\mathcal{A}}
\newcommand{\BPS}{\mathcal{U}}
\DeclareMathOperator*{\argmax}{arg\,max}

\begin{document}
\title{\huge Low-Latency Belief-Propagation Decoding of Quantum LDPC Codes with Reinforcement-Learned Scheduling\vspace{-6mm}}

\author{
Mohsen~Moradi$^*$,
Ghilas~Aissou$^*$,
Taejoon~Kim$^*$,
R\'emi~A.~Chou$^\dag$,
and David~G.~M.~Mitchell$^\ddagger$\\
$^*$School of Electrical, Computer and Energy Engineering, Arizona State University, Tempe, AZ 85287, USA\\
(Corresponding author e-mail: mmorad11@asu.edu),\\
$^\dag$Department of Computer Science and Engineering, University of Texas at Arlington, Arlington, TX 76019, USA\\ 
$^\ddagger$Klipsch School of Electrical and Computer Engineering, New Mexico State University, Las Cruces, NM, USA\vspace{-3mm}
}

\maketitle

\begin{abstract}
Belief propagation (BP) is attractive for quantum low-density parity-check (QLDPC) decoding because of its low complexity and low latency, but short cycles and degeneracy can limit its error-correction performance. Reinforcement-learned sequential scheduling (RL-S) substantially improves this performance by adapting the variable node (VN) update order, although Q-table storage and heap-based selection introduce additional overhead. We propose bucket RL-S (B-RL-S), which quantizes learned Q-values into integer priorities and schedules VN updates using dynamically maintained first-in, first-out buckets. Without retraining or changing the BP message updates, B-RL-S reduces policy storage and refreshes only priorities affected by local residual syndrome changes. 
A fixed number of priority levels enables constant-time VN selection and queue relocation. 
For the $[[288,12,18]]$ bivariate bicycle code over the
depolarizing channel at $p=0.05$, B-RL-S with eight
priority levels achieves approximately one fifth the
block error rate (BLER) of flooding BP followed by
order-10 ordered-statistics decoding. The respective
caps are $1000$ sequential iterations for B-RL-S and
$1000$ flooding iterations for the BP stage of the
reference decoder.
These results highlight strong error correction with inexpensive scheduling and no additional post-processing stage.
\end{abstract}


\section{Introduction}
Quantum low-density parity-check (QLDPC) codes offer
a promising route toward low-overhead quantum error
correction by combining sparse stabilizer checks
with the ability to protect multiple logical
qubits~\cite{bravyi2024}. Their sparse Tanner graphs
make belief propagation (BP) an attractive decoding
approach, but short cycles and degeneracy can lead
to oscillations and nonconvergence~\cite{poulin2008,
kuo2022}. A key objective is to improve BP's
error-correction performance while limiting
computational work and decoding latency. Several approaches improve BP decoding.
Ordered-statistics decoding (OSD) supplements BP
with algebraic reprocessing, while localized
statistics decoding (LSD) reduces this overhead
through localized, parallelizable
inversion~\cite{panteleev2021,roffe2020,hillmann2025}.
Both retain a separate algebraic decoding stage.
Layered and sequential schedules provide another
approach~\cite{kuo2020,ducrest2023,moradi2026seq}.


Reinforcement learning (RL) has been used to learn
message update schedules for classical
error-correcting codes~\cite{habib2021,reldec2023,
moradi2025enhancing}. Our previous work developed
reinforcement-learned sequential scheduling
(RL-S), a variable node (VN) scheduler driven by
local residual syndrome patterns~\cite{moradi2026rl}.
Its trained Q-table ranks eligible VN updates,
while incremental state maintenance and a binary
max-heap avoid repeated global searches. These
local states depend on hard-decision residuals
and change only when the corresponding residual
checks change. Full precision policy storage and
heap maintenance nevertheless introduce overhead.

We propose bucket RL-S (B-RL-S), which converts
a fixed RL-S Q-table into $K$ globally comparable
integer priority levels and schedules eligible
VNs using dynamically maintained first-in,
first-out (FIFO) buckets. B-RL-S inherits the
RL-S state representation and coupled BP updates
without retraining. It retains all active VNs
not yet updated in the current iteration and
refreshes only priorities affected by local state
changes. We specify the inference rule, analyze
scheduling work and policy storage, and bound
the instantaneous Q-value gap relative to
greedy selection. Finite-precision BP and
quantization-based decoder diversity have also
been investigated for QLDPC
codes~\cite{garciaherrero2026diversity}; here,
quantization is applied to the scheduling policy,
while the BP message representation is retained.

For the $[[288,12,18]]$ bivariate bicycle (BB)
code with depolarizing probability $p=0.05$, B-RL-S with $K=8$ achieves
approximately one fifth the block error rate
(BLER) of BP followed by order-10 OSD
(BP-OSD-10). The respective caps are $1000$
sequential iterations for B-RL-S and $1000$
flooding iterations for the BP stage of
BP-OSD-10. B-RL-S also closely tracks RL-S with
$K=4$ on the $[[144,12,12]]$ BB code and with
$K=8$ on the A5 code. For fixed $K$, the bucket
scheduler permits constant-time VN selection
and queue relocation.

\section{Background}

A Calderbank--Shor--Steane (CSS) QLDPC code~\cite{calderbank1996good,steane1996multiple}
is specified by sparse binary matrices $H_X$ and $H_Z$ satisfying
$H_X H_Z^{\mathsf{T}}=0$ over $\mathbb{F}_2$.
Under independent depolarizing noise, the Pauli error
$E_i\in\{I,X,Y,Z\}$ on qubit $i$ has the prior distribution
\begin{equation}
    \pi(q)=\Pr(E_i=q)=
    \begin{cases}
        1-p, & q=I,\\
        p/3, & q\in\{X,Y,Z\}.
    \end{cases}
\end{equation}
Define the binary error components by
\[
    e_i^X=\mathbf{1}\{E_i\in\{X,Y\}\},
    \qquad
    e_i^Z=\mathbf{1}\{E_i\in\{Z,Y\}\}.
\]
With perfect syndrome measurements, the measured syndromes are
\begin{equation}
    \bm{s}^X=H_X\bm{e}^Z,
    \qquad
    \bm{s}^Z=H_Z\bm{e}^X,
    \label{eq:syndromes}
\end{equation}
where all binary vector and matrix operations are over $\mathbb{F}_2$.
For current estimates $\hat{\bm{e}}^X$ and $\hat{\bm{e}}^Z$,
define the residual syndromes and their total Hamming weight as
\begin{equation}
    \begin{aligned}
        \bm{\delta}^X &= \bm{s}^X + H_X\hat{\bm{e}}^Z,\\
        \bm{\delta}^Z &= \bm{s}^Z + H_Z\hat{\bm{e}}^X,\\
        w &= \|\bm{\delta}^X\|_1+\|\bm{\delta}^Z\|_1.
    \end{aligned}
    \label{eq:residual}
\end{equation}
Here, $w$ counts the unsatisfied checks, and $w=0$ indicates
syndrome consistency. Quantum correction succeeds if and only if
\[
    \begin{aligned}
        \bm{e}^X+\hat{\bm{e}}^X
            &\in\operatorname{row}(H_X),\\
        \bm{e}^Z+\hat{\bm{e}}^Z
            &\in\operatorname{row}(H_Z).
    \end{aligned}
\]
A syndrome consistent correction that violates either condition
results in a logical error. BLER includes both logical errors and nonconvergence.

\subsection{Sequential Updates and Learned Priorities}
For $\alpha\in\{X,Z\}$, let $v_i$ denote the VN corresponding
to physical qubit $i$ and $c$ a check node (CN) in the
Tanner graph associated with $H_\alpha$.
The sets $\cN_\alpha(v_i)$ and $\cN_\alpha(c)$ contain
the CNs adjacent to $v_i$ and the VNs adjacent to $c$,
respectively.
The set $\cV$ consists of active VNs, each having at least one
incident edge in either graph. Within each iteration, every active
VN is updated at most once, with early termination when $w=0$.

We retain the coupled scalar message BP update of
RL-S~\cite{moradi2026rl}, based on quaternary BP~\cite{kuo2020}.
A complete VN update, $\BPS_i$, updates the incoming messages
at $v_i$ in both graphs, computes the outgoing messages using
the joint Pauli prior, and updates the quaternary hard decision
from the resulting beliefs. The joint prior accounts for the
dependence between the binary $X$- and $Z$-error components.
The message update equations, clipping rules, and cached
check product operations remain unchanged; B-RL-S modifies
only the VN selection rule.

The maximum VN degree across the two graphs is
$A=\max_{i,\alpha}|\cN_\alpha(v_i)|$.
For each VN, the neighboring residual-syndrome bits in each
graph follow the same fixed indexing used during training
and are padded with zeros to length $A$. The resulting bits
$b_{i,j}^{\alpha}$, with $j=0,\ldots,A-1$, determine the local
state index
\begin{equation}
    \sigma_i
    = \sum_{j=0}^{A-1} 2^j b_{i,j}^{X}
    + 2^A \sum_{j=0}^{A-1} 2^j b_{i,j}^{Z}.
    \label{eq:state}
\end{equation}
The trained Q-table stores a score $Q(\sigma,i)$ for each
of the $S=2^{2A}$ state indices and $n$ VNs, giving $nS$
entries. At each step, RL-S selects a VN with the largest
current score among the eligible candidates as
\begin{equation}
    i^\star \in \argmax_{i\in\cR} Q(\sigma_i,i),
    \label{eq:rls}
\end{equation}
where $\cR\subseteq\cV$ is the set of VNs not yet updated
in the current iteration. These scores serve as learned
scheduling priorities based on local residual-syndrome
patterns. The Q-table remains fixed throughout
decoding.

\section{Bucket RL-S}
An $[[n,k,d]]$ code encodes $k$ logical qubits into $n$ physical qubits and has minimum distance $d$.
B-RL-S converts (offline) the trained Q-table into $K$ integer priority
levels, where $K\geq2$. The values $q_{\min}$ and $q_{\max}$
are the minimum and maximum entries over the entire table.
For $q_{\max}>q_{\min}$, the quantization step and priority
table are
\begin{equation}
    \begin{aligned}
        \Delta
        &= \frac{q_{\max}-q_{\min}}{K},\\
        P(\sigma,i)
        &= \min\!\left\{
            K-1,\,
            \left\lfloor
                \frac{Q(\sigma,i)-q_{\min}}{\Delta}
            \right\rfloor
        \right\}.
    \end{aligned}
    \label{eq:quantize}
\end{equation}
Each entry of $P$ belongs to $\{0,\ldots,K-1\}$, with the
upper endpoint $q_{\max}$ assigned to level $K-1$.
For a constant Q-table, all priorities are zero.
The same assignment defines the $K=1$ case, which provides
a fixed-order scheduling control.

The conversion is performed offline for each trained table
and choice of $K$, without retraining. Only $P$ is retained
as the scheduling policy during decoding; neither Q-value
updates nor quantization arithmetic are required online. A common quantizer is applied across all VNs. A larger
Q-value can never receive a lower integer priority, although
different Q-values may map to the same level. 
The quantization range is determined once from the minimum
and maximum values in the entire trained Q-table and remains
fixed for all frames.
Consequently, $P$ depends only on the trained Q-table and $K$.

\subsection{Dynamic Buckets and Tie Handling}

The scheduler maintains $K$ first-in, first-out (FIFO)
queues, $\cB_0,\ldots,\cB_{K-1}$. Bucket $\cB_k$ contains the active VNs with priority $k$
that have not yet been updated in the current iteration.
At the start of each iteration, $\cR=\cV$, and the buckets
are initialized by inserting active VNs in increasing
index order. Each VN $i$ is appended to the tail of
$\cB_{P(\sigma_i,i)}$.

The next VN is selected from the head of the highest-priority
nonempty bucket, where
\begin{equation}
    \begin{aligned}
        k^\star
        &= \max\{k:\cB_k\neq\varnothing\},\\
        i
        &= \operatorname{head}(\cB_{k^\star}).
    \end{aligned}
    \label{eq:select}
\end{equation}
The selected VN is removed from its bucket and from $\cR$
before the BP update $\BPS_i$ is applied. It remains
ineligible until the next iteration.
All active VNs not yet updated in the current iteration remain
eligible for selection, including those with zero local state or
zero priority.

After a local state change, an eligible VN is moved to
the tail of its destination bucket only if its integer
priority changes. Otherwise, it retains its current
queue position. Thus, ties are resolved by FIFO order
within each priority level. The buckets determine a
serial update schedule; VNs in the same bucket are not
updated simultaneously.

\subsection{Local Maintenance After a Hard-Decision Change}

A change in the $Z$ component
flips the residual checks adjacent to $v_i$ in $H_X$,
whereas a change in the $X$ component flips those in $H_Z$.
The affected check sets are therefore
\[
    \begin{aligned}
        \cC_X(i)
        &=
        \begin{cases}
            \cN_X(v_i), & f_i^Z=1,\\
            \varnothing, & f_i^Z=0,
        \end{cases}\\
        \cC_Z(i)
        &=
        \begin{cases}
            \cN_Z(v_i), & f_i^X=1,\\
            \varnothing, & f_i^X=0.
        \end{cases}
    \end{aligned}
\]
The indicators $f_i^X$ and $f_i^Z$ equal one when the
corresponding hard-decision components change during
$\BPS_i$, and zero otherwise. 
Only VNs in
\begin{equation}
    \cA_i
    =
    \bigcup_{\alpha\in\{X,Z\}}
    \bigcup_{c\in\cC_\alpha(i)}
    \cN_\alpha(c)
    \label{eq:affected}
\end{equation}
can have a changed local state. As in RL-S~\cite{moradi2026rl}, the residual syndromes,
their total Hamming weight $w$, and the affected local
states are maintained incrementally. Each flipped
residual check toggles its corresponding bit in the state
of every incident VN, including VNs already updated in
the current iteration. When several flipped checks affect
the same VN, all corresponding state bits are updated.

After the changes in both graphs have been applied,
the priority of each distinct VN in $\cA_i\cap\cR$ is
refreshed once using its updated state. These VNs are
processed in first encounter order under a fixed traversal:
the $X$ graph is processed before the $Z$ graph, and the neighbors of each node are processed in the same predetermined order.
This determines the insertion order when several VNs
move into the same bucket. If neither hard-decision component changes, the
residual syndrome states remain unchanged and no priority
refresh is required, even though BP messages may have
changed. Algorithm~\ref{alg:bucket} summarizes the decoder.

\begin{algorithm}[t]
\caption{B-RL-S decoding with a fixed priority table}
\label{alg:bucket}
\begin{algorithmic}[1]
\small
\Require $H_X,H_Z$, syndromes $\bm{s}^X,\bm{s}^Z$, channel
         parameter $p$, priority table $P$, number of levels
         $K$, and maximum number of iterations $T$
\Ensure Pauli estimate and a Boolean indicating syndrome consistency

\State Initialize the RL-S BP messages and hard decisions
\State Compute $\bm{\delta}^X,\bm{\delta}^Z,w,$ and all $\sigma_i$

\For{$t=1,\ldots,T$}
    \State $\cR\gets\cV$; empty all buckets
    \For{$i\in\cV$ in increasing index order}
        \State Append $i$ to $\cB_{P(\sigma_i,i)}$
    \EndFor

    \While{$\cR\neq\varnothing$}
        \If{$w=0$}
            \State \Return current Pauli estimate and $\mathrm{true}$
        \EndIf

        \State Select $i$ using \eqref{eq:select}
        \State Remove $i$ from its bucket and from $\cR$
        \State Apply $\BPS_i$ and identify $f_i^X,f_i^Z$
        \State Form $\cC_X(i),\cC_Z(i)$ and $\cA_i$
        \State Update the residuals, $w$, and all states in $\cA_i$

        \For{each distinct $u\in\cA_i\cap\cR$ in first-encounter order}
            \State $k\gets P(\sigma_u,u)$
            \If{$u$ is not currently in $\cB_k$}
                \State Remove $u$ from its current bucket
                \State Append $u$ to $\cB_k$
            \EndIf
        \EndFor

        \If{$w=0$}
            \State \Return current Pauli estimate and $\mathrm{true}$
        \EndIf
    \EndWhile
\EndFor

\State \Return current Pauli estimate and $(w=0)$
\end{algorithmic}
\end{algorithm}

\subsection{Example and Score Resolution}

At the start of an iteration, consider four eligible VNs with
Q-scores $q_1=0.81$, $q_2=0.94$, $q_3=0.60$, and
$q_4=0.18$, and a global quantization range $[0,1]$. With $K=4$, the bin width is $\Delta=0.25$.
Looking up their current local states in
the stored integer priority table gives
\[
  \bigl(P(\sigma_1,1),P(\sigma_2,2),
        P(\sigma_3,3),P(\sigma_4,4)\bigr)=(3,3,2,0).
\]
Insertion in increasing VN index order gives
$\cB_3=(1,2)$, $\cB_2=(3)$, $\cB_1=\varnothing$, and
$\cB_0=(4)$, as shown in Fig.~\ref{fig:bucket-example}(a).
The scheduler selects VN~1 from the head of $\cB_3$
and removes it from eligibility before applying its BP
update.

Suppose that the updated hard decision at VN~1 flips
residual checks that change the local states of VNs~2
and~3, while VN~4 is unaffected and the residual syndrome
remains nonzero. After all state changes in both graphs
have been applied, let the updated table lookups give
\[
  P(\sigma_2',2)=3,\qquad P(\sigma_3',3)=3.
\]
The table itself remains fixed; the updated states select
its entries. VN~2 retains its position because its
priority is unchanged. VN~3 moves from $\cB_2$ to the
tail of $\cB_3$, giving $\cB_3=(2,3)$, as shown in
Fig.~\ref{fig:bucket-example}(b). Thus, VN~2 is selected
next. VN~1 remains ineligible until the next iteration,
and VN~4 remains eligible in $\cB_0$ despite its zero
priority.

If the same trained Q-table is instead quantized offline
with $K=2$, the initial priorities become $(1,1,1,0)$.
The first three VNs then occupy the same bucket and are
ordered by FIFO. The Q-score bound below analyzes the
original trained table; decoder inference accesses only
$P$.

\begin{figure}[t]
\centering
\begin{tikzpicture}[x=1cm,y=1cm,
  every node/.style={font=\footnotesize},
  token/.style={draw=black!65,line width=0.35pt,
    rounded corners=1pt,minimum width=0.86cm,
    minimum height=0.36cm,inner sep=1pt},
  selected/.style={token,line width=0.95pt,fill=black!9},
  moved/.style={token,dashed,fill=black!4},
  link/.style={-{Stealth[length=1.2mm]},line width=0.45pt}]
  \foreach \x in {0,4.20}{
    \begin{scope}[shift={(\x,0)}]
      \node[font=\scriptsize] at (2.17,2.34)
        {head $\longrightarrow$ tail};
      \foreach \k in {0,...,3}{
        \pgfmathsetmacro{\row}{0.25+0.55*\k}
        \node at (0.26,\row) {$\cB_{\k}$};
        \draw[black!30,rounded corners=1pt]
          (0.67,\row-0.235) rectangle (3.88,\row+0.235);
      }
      \node[token] at (1.35,0.25) {$v_4$};
      \node[font=\scriptsize,text=black!55] at (2.17,0.80)
        {empty};
    \end{scope}
  }
  \node at (2.02,2.72) {(a) Before selection};
  \node at (6.22,2.72) {(b) After updating $v_1$};

  \node[selected] (a1) at (1.35,1.90) {$v_1$};
  \node[token] (a2) at (2.94,1.90) {$v_2$};
  \draw[link] (a1)--(a2);
  \node[token] at (1.35,1.35) {$v_3$};

  \node[selected] (b2) at (5.55,1.90) {$v_2$};
  \node[moved] (b3) at (7.14,1.90) {$v_3$};
  \draw[link] (b2)--(b3);
  \node[font=\scriptsize,text=black!55] at (6.37,1.35)
    {empty};

  \node at (2.02,-0.20) {Select $v_1$};
  \node at (6.22,-0.20) {Next: $v_2$};
  \draw[black!25] (0.04,-0.47)--(8.08,-0.47);
  \node at (4.06,-0.71)
    {$v_2:\ 3\!\to\!3$, keep position;
     $\qquad v_3:\ 2\!\to\!3$, append.};
\end{tikzpicture}
\caption{Integer priority FIFO scheduling with $K=4$.
Tokens identify VNs; bucket labels indicate their current
integer priorities. Thick outlines mark the next VN selected,
and the dashed outline marks VN~3 after relocation.
VN~1 is removed before its BP update and remains
ineligible for the rest of the iteration.}
\label{fig:bucket-example}
\end{figure}
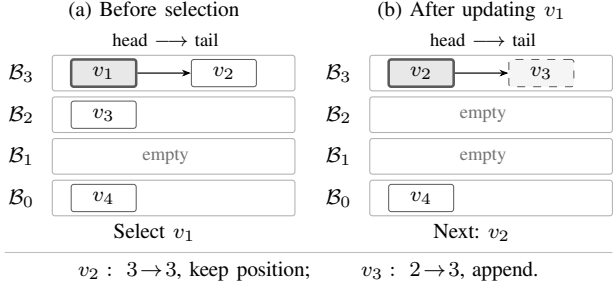
The following proposition bounds the Q-score difference
between bucket and greedy selection at a fixed decoder state.

\begin{proposition}
\label{prop:gap}
For $K\geq2$ and $q_{\max}>q_{\min}$, consider a fixed
decoder state with $\cR\neq\varnothing$.
If B-RL-S selects VN~$a$ according to
\eqref{eq:quantize} and \eqref{eq:select}, then
\begin{equation}
    0\leq\max_{i\in\cR}Q(\sigma_i,i)-Q(\sigma_a,a)
    \leq\Delta.
    \label{eq:gap}
\end{equation}
\end{proposition}

\begin{proof}
The quantizer is nondecreasing, so a maximum score
eligible VN and the selected VN belong to the same
highest occupied bucket. Its Q-score interval has width
at most $\Delta$, including the upper endpoint of the
final bin, which gives the stated bound.
\end{proof}

For a constant Q-table, the gap is zero. If a unique
maximum exceeds every other eligible score by more than
$\Delta$, both rules select the same VN at that state.
The bound concerns one selection, not accumulated reward
or BLER: different updates can change subsequent BP
messages and local states.

\section{Scheduling Cost and Policy Storage}

Table~\ref{tab:scheduling} compares the online scheduling
operations of heap-based RL-S and B-RL-S. A priority
refresh includes a table lookup and, if the priority
changes, a heap repair or bucket relocation. Although
both table lookups take $O(1)$ time, heap extraction
and priority repair require $O(\log n)$ time. Bottom-up heap
construction costs $O(n)$, whereas clearing the buckets
and inserting the eligible VNs costs $O(n+K)$.

\begin{table}[t]
\centering
\caption{Worst-case online scheduling operation
costs for fixed $K$.}
\label{tab:scheduling}
\footnotesize
\setlength{\tabcolsep}{4pt}
\renewcommand{\arraystretch}{1.12}
\begin{tabular}{@{}lcc@{}}
\toprule
Operation & RL-S & B-RL-S\\
\midrule
Initialize one iteration & $O(n)$ & $O(n+K)$\\
Select and remove a VN & $O(\log n)$ & $O(1)$\\
Refresh one VN priority & $O(\log n)$ & $O(1)$\\
\bottomrule
\end{tabular}
\vspace{-3mm}
\end{table}

Let $T$ be the iteration cap, $I\leq T$ the number of
iterations actually started, $U$ the total number of VN
updates, and $R$ the total number of affected eligible VN
priority refreshes. An affected eligible VN is counted
at most once per VN update, even if several of its state
bits change, and can be counted again after a later
update. These counts can differ between decoders.

Each initialized iteration requires one scheduler build,
each VN update requires one selection and removal,
and each affected eligible VN requires one priority
refresh. Multiplying the operation costs in
Table~\ref{tab:scheduling} by these counts gives
\begin{equation}
\begin{aligned}
C_{\mathrm{heap}}
  &=I\,O(n)+U\,O(\log n)+R\,O(\log n)\\
  &=O\!\left(In+(U+R)\log n\right),\\
C_{\mathrm{bucket}}
  &=I\,O(n+K)+U\,O(1)+R\,O(1)\\
  &=O\!\left(I(n+K)+U+R\right).
\end{aligned}
\label{eq:scheduling-cost}
\end{equation}
Each iteration updates at most $n$ VNs, giving
$U\leq In\leq Tn$. With bounded VN and CN degrees,
one update affects at most a constant number $c$ of
eligible VNs, so $R\leq cU$. Thus, for fixed $K$,
initialization contributes $O(Tn)$ to both schedulers,
while selection and priority refreshes contribute
$O(Tn\log n)$ for the heap and $O(Tn)$ for the buckets.
The worst-case scheduling cost therefore decreases
from $O(Tn\log n)$ to $O(Tn)$.

For comparison, flooding BP updates all edge messages
and hard decisions in each iteration. If $E$ is the
total number of edges in both Tanner graphs, these
operations require $O(E+n)$ work per iteration.
Bounded degrees give $E=O(n)$ and hence $O(Tn)$ work
over at most $T$ iterations. The BP updates and
incremental residual-state maintenance excluded from
\eqref{eq:scheduling-cost} contribute $O(U)$ work under
the same degree bounds. Including them, B-RL-S also
has $O(Tn)$ worst-case decoding work while retaining
an adaptive sequential schedule.

For an idealized latency comparison, assume that fully
parallel flooding BP completes an entire iteration,
including hard decisions and syndrome testing, in one
clock cycle. Flooding then requires at most $T$ cycles.
Under the same degree bounds and fixed $K$, assume
constant-cycle local VN processing and indexed memory
accesses. Random sequential BP and B-RL-S then require
$O(Tn)$ cycles, whereas RL-S with a serial memory-based
heap requires $O(Tn\log n)$ cycles because each heap
traversal takes $O(\log n)$ cycles. 
Cluster-based scheduling~\cite{moradi2026cluster} offers
further reductions in sequential BP update depth, with
conflict-free color clustering~\cite{moradi2026color}
requiring $O(T\chi_{\mathrm{VN}})$ parallel update stages,
where $\chi_{\mathrm{VN}}$ is the number of color classes,
assuming constant-cycle parallel processing within each
class and excluding scheduling overhead.

\section{Numerical Results}
\label{sec:results}

\subsection{Simulation Settings}
We consider $[[144,12,12]]$ and $[[288,12,18]]$
BB codes~\cite{bravyi2024} and the $[[180,10,15\leq d\leq18]]$ A5 code~\cite{panteleev2021}. Errors follow the
independent depolarizing model, and syndrome measurements
are noiseless. For each code, the priority tables for all
values of $K$ are obtained from the same frozen RL-S
Q-table, without retraining. B-RL-S retains the message
updates, initialization, and syndrome-based stopping rule
of RL-S. The component-prior log-likelihood ratio (LLR) is
$\mu=\log[(1-2p/3)/(2p/3)]$, and message magnitudes are
clipped at~25.

We report BLER, counting both nonconvergence within the
prescribed cap and syndrome consistent logical errors.
BP and BP-OSD-10 are included as references, with OSD
order~10. For RL-S and B-RL-S, $T$ is the maximum number
of sequential iterations, each containing at most $n$ VN
updates. For flooding BP and the BP stage of BP-OSD,
$T$ is the maximum number of flooding iterations.

\subsection{Error-Correction Performance}

\begin{figure}[t]
  \centering
  \includegraphics[width=0.85\linewidth]{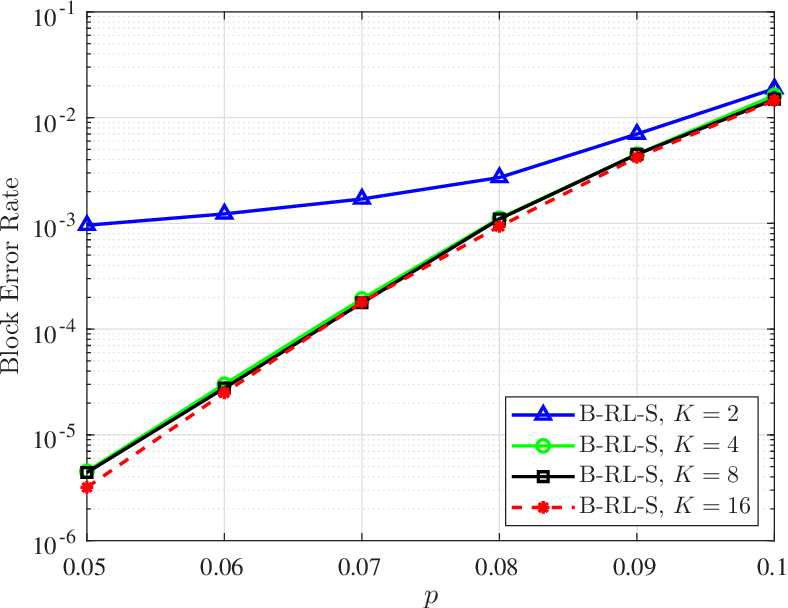}
  \caption{BLER of B-RL-S on the $[[288,12,18]]$ BB code
  over the depolarizing channel for $K\in\{2,4,8,16\}$.\vspace{-4mm}}
  \label{fig:Quantized_RL_S_BB288_FER}
\end{figure}

Fig.~\ref{fig:Quantized_RL_S_BB288_FER} shows the effect
of priority resolution on the $[[288,12,18]]$ BB code. 
With $K=2$, BLER is
substantially larger at small $p$ than with the finer
quantizers. In contrast, the curves for $K=4$, $8$, and
$16$ remain close over the plotted range. 
For this experiment, increasing the
resolution beyond four levels provides little further
improvement in the plotted BLER. 

\begin{figure}[t]
  \centering
  \includegraphics[width=0.85\linewidth]{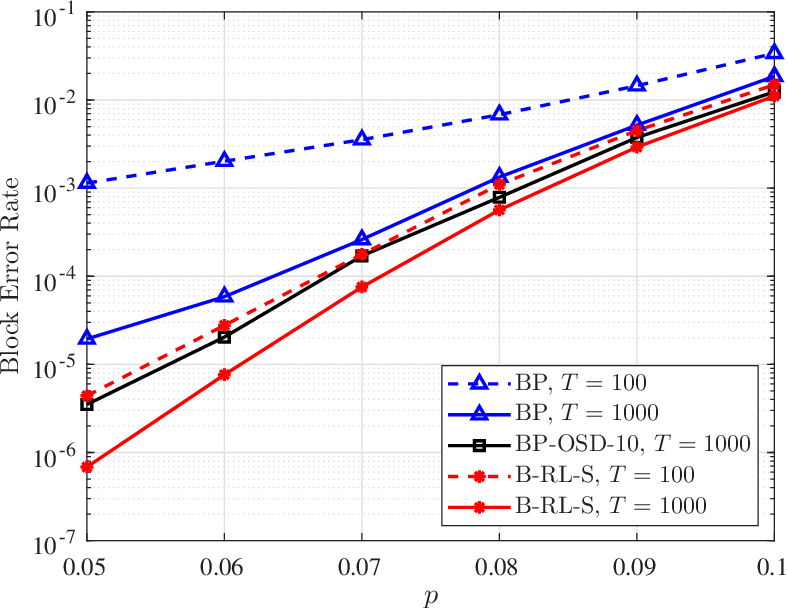}
  \caption{BLER on the $[[288,12,18]]$ BB code over the depolarizing channel:
  B-RL-S with $K=8$ and $T\in\{100,1000\}$, BP with
  $T\in\{100,1000\}$, and BP-OSD-10 with a BP cap of
  $T=1000$.\vspace{-4mm}}
  \label{fig:BB288_QRL_LSvOSD}
\end{figure}

Fig.~\ref{fig:BB288_QRL_LSvOSD} compares B-RL-S with
$K=8$ against BP and BP-OSD-10. Increasing the B-RL-S
iteration cap from $T=100$ to $T=1000$ lowers BLER throughout
the plotted range. At $T=1000$, B-RL-S outperforms both
displayed baselines. In particular, at $p=0.05$, its
BLER is approximately five times lower than that of
BP-OSD-10 with $T=1000$, without an OSD post-processing
stage. 
With $T=100$, B-RL-S also achieves a BLER close
to that of BP-OSD-10 with a BP cap of $1000$
iterations.

\begin{figure}[t]
  \centering
  \includegraphics[width=0.85\linewidth]{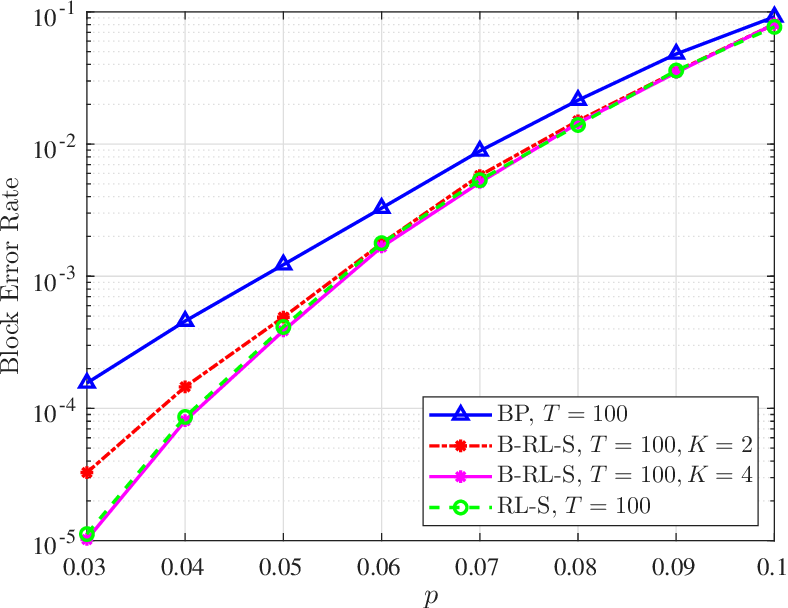}
  \caption{BLER on the $[[144,12,12]]$ BB code over the
  depolarizing channel: B-RL-S with $K\in\{2,4\}$,
 RL-S, and BP, all with $T=100$.\vspace{-4mm}}
  \label{fig:Quantized_RL_S_v_BP_144_FER}
\end{figure}

Fig.~\ref{fig:Quantized_RL_S_v_BP_144_FER} compares
B-RL-S with RL-S on the $[[144,12,12]]$ BB code at $T=100$.
Four priority levels closely track RL-S across the
plotted range, while two levels incur a more visible
loss at the smallest error probabilities. Both
quantized schedules improve on the displayed BP curve.

\begin{figure}[t]
  \centering
  \includegraphics[width=0.85\linewidth]{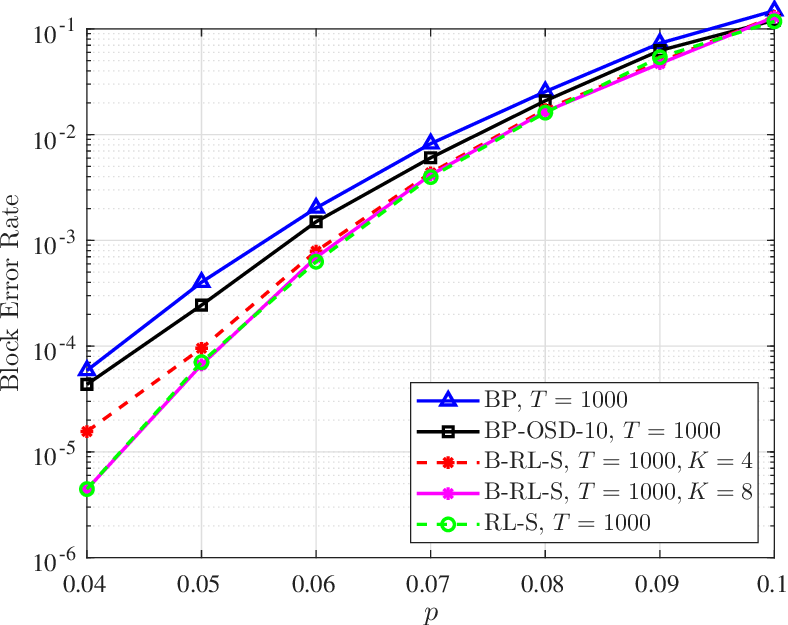}
  \caption{BLER on the $[[180,10,15\leq d\leq18]]$ A5 code over the depolarizing channel:
  B-RL-S with $K\in\{4,8\}$, RL-S, BP, and BP-OSD-10,
  all with $T=1000$.\vspace{-4mm}}
  \label{fig:Quantized_RL_S_v_OSD_A5_FER}
\end{figure}

For the A5 code at $T=1000$,
Fig.~\ref{fig:Quantized_RL_S_v_OSD_A5_FER} shows close
agreement between $K=8$ and RL-S. With $K=4$, the
performance loss is again most evident at small $p$,
although the curve remains below the displayed BP and
BP-OSD-10 curves. Together, these results show that a
small priority alphabet can preserve much of the RL-S
gain, with the required resolution depending on the
code and operating point.

\section{Conclusion}
Our proposed B-RL-S replaces full-precision learned scheduling scores
with integer priorities and dynamically maintained FIFO
buckets, while retaining the states and BP updates of
RL-S. Numerical results show that a small number of
priority levels can preserve much of the BLER improvement achieved by RL-S. On the $[[288,12,18]]$ BB code, the $K=8$ decoder achieves an
approximately fivefold lower BLER than BP-OSD-10 at
depolarizing probability $p=0.05$ with (worst-case) iteration caps of $1000$. Indexed bucket operations
reduce scheduling overhead. These benefits require no retraining or
algebraic post-processing.

\section{Acknowledgment}
This work is in part supported by the National Science Foundation (NSF) under grants CNS2451268, CNS2514415, CCF2145917, and ITE2515378, and the Office of Naval Research (ONR) under Grant N000142112472. 

\bibliographystyle{IEEEtran}
\bibliography{bibliography}

@article{bravyi2024,
  author  = {Bravyi, S. and Cross, A. W. and Gambetta, J. M. and
             Maslov, D. and Rall, P. and Yoder, T. J.},
  title   = {High-threshold and low-overhead fault-tolerant quantum memory},
  journal = {Nature},
  volume  = {627},
  pages   = {778--782},
  year    = {2024},
  doi     = {10.1038/s41586-024-07107-7},
  url     = {https://doi.org/10.1038/s41586-024-07107-7}
}

@article{poulin2008,
  author        = {Poulin, D. and Chung, Y.},
  title         = {On the iterative decoding of sparse quantum codes},
  journal       = {Quantum Inf. Comput.},
  volume        = {8},
  number        = {10},
  pages         = {987--1000},
  year          = {2008},
  eprint        = {0801.1241},
  archivePrefix = {arXiv},
  url           = {https://arxiv.org/abs/0801.1241}
}

@article{kuo2022,
  author  = {Kuo, K.-Y. and Lai, C.-Y.},
  title   = {Exploiting degeneracy in belief propagation decoding of quantum codes},
  journal = {npj Quantum Inf.},
  volume  = {8},
  pages   = {111},
  year    = {2022},
  doi     = {10.1038/s41534-022-00623-2},
  url     = {https://doi.org/10.1038/s41534-022-00623-2}
}

@article{panteleev2021,
  author  = {Panteleev, P. and Kalachev, G.},
  title   = {Degenerate quantum {LDPC} codes with good finite length performance},
  journal = {Quantum},
  volume  = {5},
  pages   = {585},
  year    = {2021},
  doi     = {10.22331/q-2021-11-22-585},
  url     = {https://doi.org/10.22331/q-2021-11-22-585}
}

@article{roffe2020,
  author  = {Roffe, J. and White, D. R. and Burton, S. and Campbell, E.},
  title   = {Decoding across the quantum low-density parity-check code landscape},
  journal = {Phys. Rev. Research},
  volume  = {2},
  pages   = {043423},
  year    = {2020},
  doi     = {10.1103/PhysRevResearch.2.043423},
  url     = {https://doi.org/10.1103/PhysRevResearch.2.043423}
}

@article{hillmann2025,
  author  = {Hillmann, T. and Berent, L. and Quintavalle, A. O. and
             Eisert, J. and Wille, R. and Roffe, J.},
  title   = {Localized statistics decoding for quantum low-density parity-check codes},
  journal = {Nature Commun.},
  volume  = {16},
  pages   = {8214},
  year    = {2025},
  doi     = {10.1038/s41467-025-63214-7},
  url     = {https://doi.org/10.1038/s41467-025-63214-7}
}

@article{kuo2020,
  author        = {Kuo, K.-Y. and Lai, C.-Y.},
  title         = {Refined belief propagation decoding of sparse-graph quantum codes},
  journal       = {IEEE J. Sel. Areas Inf. Theory},
  volume        = {1},
  number        = {2},
  pages         = {487--498},
  year          = {2020},
  eprint        = {2002.06502},
  archivePrefix = {arXiv},
  url           = {https://arxiv.org/abs/2002.06502}
}

@inproceedings{ducrest2023,
  author        = {du Crest, J. and Garcia-Herrero, F. and Mhalla, M. and
                   Savin, V. and Valls, J.},
  title         = {Layered decoding of quantum {LDPC} codes},
  booktitle     = {Proc. 12th Int. Symp. Topics Coding (ISTC)},
  year          = {2023},
  pages         = {1--5},
  eprint        = {2308.13377},
  archivePrefix = {arXiv},
  url           = {https://arxiv.org/abs/2308.13377}
}

@inproceedings{moradi2026seq,
  author    = {Moradi, M. and Habib, S. and Nourozi, V. and Mitchell, D. G. M.},
  title     = {Sequential {BP}-based decoding of {QLDPC} codes},
  booktitle = {Proc. IEEE ICC},
  year      = {2026},
  pages     = {1--6},
  doi       = {10.1109/ICC59461.2026.11587995},
  url       = {https://doi.org/10.1109/ICC59461.2026.11587995}
}

@article{habib2021,
  author  = {Habib, S. and Beemer, A. and Kliewer, J.},
  title   = {Belief propagation decoding of short graph-based channel codes via reinforcement learning},
  journal = {IEEE J. Sel. Areas Inf. Theory},
  volume  = {2},
  number  = {2},
  pages   = {627--640},
  year    = {2021},
  doi     = {10.1109/JSAIT.2021.3073834},
  url     = {https://doi.org/10.1109/JSAIT.2021.3073834}
}

@article{reldec2023,
  author  = {Habib, S. and Beemer, A. and Kliewer, J.},
  title   = {{RELDEC}: Reinforcement learning-based decoding of moderate length {LDPC} codes},
  journal = {IEEE Trans. Commun.},
  volume  = {71},
  number  = {10},
  pages   = {5661--5674},
  year    = {2023},
  doi     = {10.1109/TCOMM.2023.3296621},
  url     = {https://doi.org/10.1109/TCOMM.2023.3296621}
}

@misc{moradi2026rl,
  author        = {Moradi, M. and Nourozi, V. and Habib, S. and Mitchell, D. G. M.},
  title         = {Learning to decode quantum {LDPC} codes via belief propagation},
  year          = {2026},
  eprint        = {2603.10192},
  archivePrefix = {arXiv},
  note          = {arXiv:2603.10192},
  url           = {https://arxiv.org/abs/2603.10192}
}

@ARTICLE{moradi2025enhancing,
  author={Moradi, Mohsen and Habib, Salman and Mitchell, David G. M.},
  journal={IEEE Communications Letters}, 
  title={Enhancing Belief Propagation Decoding of Polar Codes: A Reinforcement Learning Approach}, 
  year={2025},
  volume={29},
  number={6},
  pages={1285-1289},
  doi={10.1109/LCOMM.2025.3559466}
}

@article{calderbank1996good,
  author  = {Calderbank, A. R. and Shor, P. W.},
  title   = {Good quantum error-correcting codes exist},
  journal = {Phys. Rev. A},
  volume  = {54},
  number  = {2},
  pages   = {1098--1105},
  year    = {1996},
  doi     = {10.1103/PhysRevA.54.1098}
}

@article{steane1996multiple,
  author  = {Steane, A. M.},
  title   = {Multiple-particle interference and quantum error correction},
  journal = {Proc. R. Soc. Lond. A},
  volume  = {452},
  number  = {1954},
  pages   = {2551--2577},
  year    = {1996},
  doi     = {10.1098/rspa.1996.0136}
}

@misc{moradi2026cluster,
  author       = {Mohsen Moradi and Taejoon Kim and
                  R{\'e}mi A. Chou and David G. M. Mitchell},
  title        = {Learning to Decode Quantum {LDPC} Codes via
                  Cluster-Based Sequential Belief Propagation},
  howpublished = {arXiv:2607.20130},
  year         = {2026},
  url          = {https://arxiv.org/abs/2607.20130}
}

@misc{moradi2026color,
  author       = {Mohsen Moradi and Taejoon Kim and
                  R{\'e}mi A. Chou},
  title        = {Conflict-Free Color-Clustered Sequential
                  Belief-Propagation Decoding of Quantum {LDPC}
                  Codes via Reinforcement Learning},
  howpublished = {arXiv:2609.20236},
  year         = {2026},
  url          = {https://arxiv.org/abs/2609.20236}
}

@article{garciaherrero2026diversity,
  author = {Garcia-Herrero, Francisco and Valls, Javier
            and Vergara-Picazo, Llanos and Torres, Vicente},
  title = {Diversity Methods for Improving Convergence
           and Accuracy of Quantum Error Correction
           Decoders Through Hardware Emulation},
  journal = {Quantum},
  volume = {10},
  pages = {2071},
  year = {2026},
  doi = {10.22331/q-2026-04-16-2071}
}

\end{document}